\documentclass[letterpaper, 10 pt, conference]{ieeeconf}

\IEEEoverridecommandlockouts                              

\usepackage{nomencl}
\makeglossary
\usepackage{graphicx}
\usepackage{caption}
\usepackage{subcaption}
\usepackage{algorithm}
\usepackage{algorithmic}
\usepackage{amsmath}
\usepackage{graphicx}
\usepackage{enumerate}%
\usepackage{type1cm}
\usepackage{lettrine}
\usepackage{breqn}
\usepackage{graphicx}
\usepackage{amssymb}
\usepackage{framed}
\usepackage{tikz}
\usepackage{amsmath}
\usepackage{setspace}
 \usepackage{pgfplots}
\usepackage{caption}
\usepackage{subcaption}
\usepackage{mathtools}

\newtheorem{mydef}{\textnormal{\textbf{Definition}}}

\usepackage{epstopdf}

\usepackage{bm}

  \newcommand{\diag}{\mathop{\rm diag}}

\newtheorem{cond}{\textnormal{\textbf{Condition}}}

\newtheorem{problem}{\textnormal{\textbf{Problem}}}
\newtheorem{assumption}{\textnormal{\textbf{Assumption}}}

\newtheorem{lemma}{\textnormal{\textbf{Lemma}}}

\usepackage{calc}

\newcommand{\matr}[1]{\mathbf{#1}}

\newtheorem{proposition}{\textnormal{\textbf{Proposition}}}
\newtheorem{theorem}{\textnormal{\textbf{Theorem}}}

\newcommand{\nunder}[2][5]{\mathrlap{\mkern\the\numexpr#1/2mu\relax\underlipe{\phantom{\mathrm{#2}\mkern-#1mu}}}#2}
\usepackage[font={footnotesize}]{caption}
\title{\LARGE \bf
Short-term Predictable Power of a Wind Farm via Distributed Control of Wind Generators with Integrated Storage
}

\author{Stefanos~Baros$^{1}$ 
\thanks{$^{1}$Stefanos Baros is with the ECE department
of Carnegie Mellon University, Pittsburgh, PA, 15213 USA e-mail: {\tt sbaros@andrew.cmu.edu}}}

\begin{document}

\maketitle
\thispagestyle{empty}
\pagestyle{empty}

\begin{abstract}
 In this paper we introduce a Leader-follower consensus  protocol and study its stability properties with and without communication delays. On the practical side, we explore its application on coordinating a group of wind Double-Fed Induction Generators (DFIGs) with integrated storage.  To begin with, we establish asymptotic stability of the consensus protocol by employing singular perturbation theory. Subsequently, we  establish asymptotic stability of the protocol under communication delays using a Lyapunov-Krasovskii functional. Lastly, we use the proposed protocol to design a methodology that can be adopted by a fleet of state-of-the-art wind generators (WGs).  The objective is that the WGs self-organize and control their storage devices such that WF total power output is tracking a reference while equal contribution from each storage device is attained i.e equal power output from each storage. We demonstrate the effectiveness of our results and the corresponding approach via simulations on the IEEE 24-bus RT system.

\end{abstract}


\IEEEpeerreviewmaketitle

\section{Introduction}

According to a US Department of Energy (DOE) study \cite{windvision}, wind generators (WGs) are going to produce one of the largest shares of renewable energy in the future. Specifically, the same study argues that, by 2020, 10\% of the US electricity demand is expected to be produced by wind generators, offshore and onshore \cite{windvision}. Clearly, US stake their renewable energy future on wind generators for meeting their national goals regarding the reduction of the carbon dioxide emissions. A similar status is evident in most of the countries around the world. However, when power systems accommodate high-levels of wind integration they have to face a crucial challenge. That is, to guarantee their stability, their reliability and their dynamical performance which becomes highly dependent on the control methods that are adopted by the various WGs. 

 \par  The problem of coordinating and controlling a group of WGs in order for the WF power output to track a reference had been studied in \cite{constantpower},\cite{fullydistribdfig},\cite{biegel}. The authors in \cite{constantpower} introduced  a centralized two-layer constant power control scheme for a
WF comprised with wind DFIGs  with integrated storage. We emphasize  that, wind DFIGs with energy storage are the gold standard on WGs technology, able to offer increased flexibility for short-term scheduling and power output smoothing  \cite{integrationofenergystorage},\cite{millerge}. On the high-layer \cite{constantpower}, the wind farm  supervisory controller computed the power references for the low-layer controllers of the DFIGs, where on the low-layer, each DFIG controller generated power equal to the reference using the storage. It is important to underline that centralized control schemes as the one in \cite{constantpower}, require communication links from each WG  to a central wind farm controller and extensive computations centrally. As a result,  centralized approaches are plagued by expensive communication network costs, single-point failures, high computational cost and delays \cite{selforga}. In the future, the number of large WFs with hundreds of WGs is going to increase all over the world. This prediction rests on the numerous renewable energy national plans around the world that have the expected wind generation dominating the total generation share coming from the renewable energy resources. Nevertheless, in high-wind-integration settings, the disadvantages of using centralized approaches for the coordination and control of large wind farms will become more pronounced. Consider that, with centralized approaches, the performance and the reliability of a WF in providing services can be severely compromised.  On the other hand, distributed  methods for controlling WFs with state-of-the-art WGs are becoming very appealing due to the reduced communication infrastructure they require, the low computational costs and the increased reliability that characterizes them \cite{selforga}, \cite{biegel}. In fact, such methods can enable future  WFs to perform robustly and reliably toward achieving certain objectives e.g rapidly track a power reference point, providing frequency and inertial response or alleviating intermittency. Distributed approaches for addressing the problem of coordination and control of WGs were studied in \cite{fullydistribdfig},\cite{biegel}. Particularly, in \cite{fullydistribdfig}, the authors proposed a strategy, based on multi-agent theory, to coordinate the WGs in a microgrid. The total available wind power and the total demand were retrieved in a distributed way, using agents at each bus and two average consensus protocols that were executed in parallel. Subsequently, the total demand and the wind power were used to compute the set points of the WGs.  In the same spirit, the authors in \cite{biegel}, derived a distributed  controller to regulate the set-points of various WGs. The control objectives for these controllers were to generate a specific total power output and at the same time minimize the fatigue on the wind turbines.
\par In \cite{fullydistribdfig},\cite{biegel} distributed control of DFIGs with no integrated storage is concerned. Also, since regulation of the WGs set-points was aimed, only control of the Rotor-side-converter (RSC) is studied. In this work, we are focusing on the particular problem of controlling the storage devices of \textit{state-of-the-art WGs} in a fair-sharing fashion such that the WF total power matches a committed power output, that can be constant or quasi-constant. The only work related to this problem is the work in \cite{constantpower}, which considered a centralized approach. We underline that state-of-the-art WGs are expected to be widely deployed in the future. In light of that, deriving new distributed methods for their efficient control and coordination becomes both an emerging and important challenge. 

\subsection{Our Contributions} 
 Motivated by multi-agent systems theory, we make the following contributions toward addressing the problem of distributed coordination and control of state-of-the-art WGs: $a)$ We  propose a distributed leader-follower consensus protocol that requires minimal  communication links. Using singular perturbation arguments we prove that this protocol converges to its \textit{asymptotically stable} equilibrium point, under certain conditions.   $b)$ Further, we prove that the proposed protocol possesses a \textit{delay-independent stability property} which induces certain robustness to it with respect to time-delays that usually arise in communication networks $c)$ On the practical side, we adopt the consensus protocol in a WF set-up and show that it can achieve the following objectives. It can effectively coordinate the storage devices of a group of WGs such that the WF total power matches a reference while consensus on the power output of the storage devices is reached.

\par We outline the rest of the paper as follows. In Section II, we provide some preliminares and introduce the leader-follower consensus protocol.
In Section III, stability analysis of the proposed protocol is studied, using singular perturbation arguments. In the same Section, stability analysis under time-delays using Lyapunov-Krasovskii functionals is  presented in detail. In Section IV, we apply the proposed protocol in a WF, to address the problem of fair utilization of WGs storage devices for short-term power output tracking. In Section V, we validate our results with simulations on the modified IEEE-RTS 24-bus system. Finally, Section VI concludes the paper.

\section{Problem Formulation}
\label{problemform}
\subsection{Notation}

In this paper we use standard notation. We denote by $\mathbb{R}$  the set of reals and by $\mathbb{C}$ the set of complex numbers. Also, we denote by $\mathbb{R}_{+}$ the set of non-negative real numbers and with $\mathbb{R}_{++}$ the set of positive reals. The $m$-dimensional Euclidean space is denoted by $\mathbb{R}^m$. We denote vectors and matrices with bold characters. With $\matr{A}\in\mathbb{R}^{m\times n}$ we denote a $m\times n$ matrix of reals. With $\matr{A}^\top$ we denote the transpose of $\matr{A}$ and with $[a]_{ij}$ the $(i,j)$-entry of the matrix $\matr{A}$. With $\matr{A}\succ0\; (\matr{A}\succeq 0)$ we denote that the matrix $\matr{A}$  is positive definite (semi-definite). The spectrum  of the matrix $\matr{A}$ (set of eigenvalues) is denoted by $\sigma(\matr{A})$. A $n\times n$ diagonal matrix $\matr{B}$ is denoted by $\matr{B}=\diag[b_i]_{i=1}^n$. Define a $n$-dimensional vector $\matr{a}$ as $\matr{a}=[a_i]_{i=1}^n$. The operator $\|\cdot\|_2$ denotes the standard Euclidean norm $\mathcal{L}_2$ .  The maximum value of the vector $\matr{a}$ is denoted by $\bar{\matr{a}}$. 
  Similarly the maximum value of a scalar quantity $z$ is given by $\bar{z}$.  With $\matr{I}_n$ we denote the $n\times n$ identity matrix and the with $\matr{0}_{n\times 1}$ and $\matr{1}_{n\times 1}$ a $n\times 1$ column vector of zeros and ones respectively.
 With $x$ being a variable, we denote its time derivative with $\dot{x}$ ($\frac{dx}{dt}$). The operator $\operatorname{Re}(\cdot)$ returns the real part of an imaginary number $\cdot\in\mathbb{C}$.
\vspace{-3mm}
\subsection{Leader-Follower Consensus Protocol}
Consider  a set of agents denoted by  $\mathcal{V}\triangleq\{1,...,n\}$ and indexed by indexed by $i\in\mathcal{V}$. The exchange of information between the agents can be represented as a graph $\mathcal{G}=(\mathcal{V},\mathcal{E})$ where $\mathcal{V}$ is the set of nodes and 
\begin{equation}
\mathcal{E}\triangleq\Big \{(i,j)\;|\; j=i+1,\; i\in\mathcal{V}\setminus \{n\},\;j\in\mathcal{V}\setminus \{1\}\Big \},\;\;\mathcal{E}\subseteq\mathcal{V}\times\mathcal{V} \nonumber
\end{equation} is the set of edges corresponding to the allowable communication between the agents as shown in Fig.~\ref{nodegraph}. The edge $(i,j)$ denotes information exchange between agent $i$ and agent $j$. 
\begin{figure}
\centering
\includegraphics[scale=0.37]{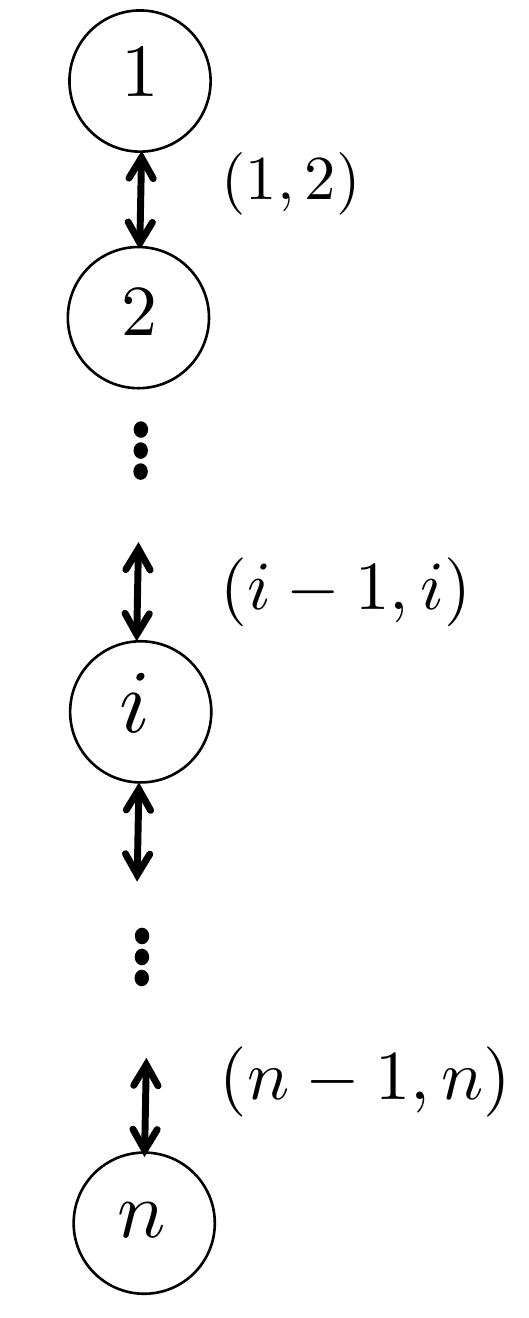}
\caption{Graph $\mathcal{G}=(\mathcal{V},\mathcal{E})$ denoting the information exchange between the agents}
\label{nodegraph}
\vspace{-6.5mm}
\end{figure}
Without loss of generality we assign as leader $l$ the agent $l\triangleq 1$. Observe that, each agent exchanges information with one neighboring agent as shown in Fig.~\ref{nodegraph}. Now, let $\bar{\mathcal{V}}$ denote the set of agents without the leader.  In our case, $\bar{\mathcal{V}}\triangleq\{2,...,n\}$ where $\bar{\mathcal{V}}\subset \mathcal{V}$. In the above set-up, consider the case where the agents desire to reach consensus on a relevant state-variable $z_i,\;i\in\mathcal{V}$. The only constraint is that the sum of their consensus states has to be $z^{*}$. For addressing this problem under the above set-up, we propose the following \textit{Leader-follower Consensus Protocol} \cite{selforga}.
\vspace{-2mm}
\begin{subequations}
\begin{align}
\intertext{\textbf{Protocol $\mathcal{P}_1$}}
\intertext{\textit{Leader agent}}
\frac{d\xi_h}{dt}&\triangleq\Big(z^*-\sum_{i\in \mathcal{V}} z_i\Big),\hspace{9mm}\xi_h\in\mathbb{R}\label{consensus_equations1}\\
\frac{dz_l}{dt}&\triangleq-k_{\alpha,l}(z_l-\xi_h)\;,\hspace{5mm}  \label{consensus_equations2}\\
z_l&\triangleq z_1,\nonumber\hspace{27mm} z_l\in\mathbb{R}
\intertext{\textit{Agent} $i$}
\frac{dz_i}{dt}&\triangleq -k_{\alpha,i}(z_i-z_{i-1})\;,\;\;\;\;\;z_i\in\mathbb{R},\;i\in\bar{\mathcal{V}} \label{consensus_equations3}
\end{align}
\end{subequations}
 The variables $k_{\alpha,l},k_{\alpha,i}\in\mathbb{R}_{++}$ are the gains of the consensus protocol, $\xi_h$ is the auxiliary state of the leader agent and $z_l$ its consensus state, whereas,  $z_i,\;\forall i \in\mathcal{\bar{V}}$, the consensus states of all the other agents. 
 Let the following assumption hold.
  \begin{assumption}
\textit{The information $\sum_{i\in\bar{\mathcal{V}}}z_i$ can be retrieved by the leader agent $l$.}
 \label{pmtotal}
 \end{assumption}
The above information can be obtained by the leader agent through indirect information passing from each agent to the leader (even without direct connection).
\section{Stability Analysis of the Leader-Follower Consensus Protocol}
\subsection{Two-time-scales Property}
 In the protocol $\mathcal{P}_1$ dynamics a two-time-scales property emerges for certain values of the gains $k_{\alpha,l},k_{\alpha,i}$.  To prove that, we show that the system given by equations \eqref{consensus_equations1}-\eqref{consensus_equations3} attains a \textit{standard singular perturbed form} \cite{khalil}. Let the vectors $\matr{z}$ and $\tilde{\matr{z}}$ be defined as
 \begin{align*}
 \matr{z}&\triangleq[z_2\; ... \;z_{n}]^\top,\;\;\matr{z}\in \mathbb{R}^{n-1}\\
  \tilde{\matr{z}}&\triangleq[z_l \; \matr{z}^\top]^\top, \hspace{6mm}\tilde{\matr{z}}\in\mathbb{R}^n
\end{align*}   respectively. Then, the following proposition holds.   
  \begin{proposition} \textnormal{The protocol $\mathcal{P}_1$ given by Eq.~\eqref{consensus_equations1}-\eqref{consensus_equations3} acquires a  two-time-scales property for $k_{a,l}=k_{a,i}\gg 1,\;\forall i$.}
  \end{proposition}
  \begin{proof}
Without loss of generality we consider that $k_{\alpha,l}=k_{\alpha,i},\;\;\forall i\in\bar{\mathcal{V}}$. Accordingly, we define the time-scale separation parameters as $\varepsilon=\frac{1}{k_{\alpha,l}}=\frac{1}{k_{\alpha,i}} \in \mathbb{R},\; \forall i\in\mathcal{V}$. With those,  equations \eqref{consensus_equations1}-\eqref{consensus_equations3} become:
 \begin{subequations}
\begin{align}
\frac{d\xi_h}{dt}&\triangleq\Big(z^*-\sum_{i\in\mathcal{V}} z_i)\label{consensus_equations7}\\
\varepsilon\frac{dz_l}{dt}&\triangleq-(z_l-\xi_h)\;,\;\; l\triangleq 1\label{consensus_equations8}\\
\varepsilon\frac{dz_i}{dt}&\triangleq-(z_i-z_{i-1})\;,\;\;i\in\bar{\mathcal{V}}\label{consensus_equations9} 
\end{align}
\end{subequations}
Then, equations \eqref{consensus_equations7}-\eqref{consensus_equations9} can be written in the next form:
  \label{singpert}
  \begin{subequations}
  \begin{align}
\frac{d{\xi}_h}{dt}&\triangleq g_h
\label{standard_sing_pertu1}
\\\varepsilon \frac{\matr{d\tilde{z}}}{\matr{d}\bm{t}}&\triangleq\matr{\tilde{g}}
\label{standard_sing_pertu2}
  \intertext{where}  
  \frac{\matr{d{\tilde{z}}}}{\matr{d}\bm{t}}&\triangleq[\frac{d{z}_l}{dt}\;...\;\frac{d{z}_i}{dt}\; ...\; \frac{d{z}_n}{dt}]^\top\; \label{dzdt}\\
g_h&\triangleq (z^*-\sum_{i\in\mathcal{V}} z_i) \in \mathbb{R}\label{gh}\\
  \matr{\tilde{g}}&\triangleq-[(z_l-\xi_h)...(z_n-z_{n-1})]^\top\label{gmatr}
  \end{align}
  \end{subequations}
 with $\frac{\matr{d{\tilde{z}}}}{\matr{d}t},\matr{\tilde{g}} \in \mathbb{R}^{n}$ being smooth vector fields and  $\varepsilon$ the corresponding small positive parameter. The equations \eqref{standard_sing_pertu1},\eqref{standard_sing_pertu2} are in the standard singularly perturbed form with two time scales, namely $t$ and $\tau=\frac{t}{\varepsilon}$. That, completes the proof.
  \end{proof} 
\subsection{Asymptotic Stability}
In this section, we rely on the two-time-scales property stated above  and singular perturbation theory  for proving asymptotic stability of the equilibrium point of $\mathcal{P}_1$ in a compartmental fashion. The above analysis reveals that the \textit{slow quasi-steady state} is the augmented state-variable $\xi_h$ of the leader agent $l$, and the \textit{fast boundary-layer states}, are the consensus state-variables $\matr{\tilde{z}}$.  We first establish the stability of the equilibrium point of the fast sub-system and of the equilibrium point of the slow sub-system, when these systems are considered decoupled from each other. Succeeding that step, we use \textit{Theorem 11.4} \cite{khalil} to directly deduce that a maximum bound $\varepsilon^*$ on the small parameter $\varepsilon$ exists under which stability of the full coupled system is guaranteed. The \textit{equilibrium point of} \eqref{consensus_equations7}-\eqref{consensus_equations9} is:
\begin{subequations}
\begin{align}
 \xi_{h0}&=z^*/n\vspace{-2mm}\label{equil1}\\
z_{l0}&=\xi_{h0}\label{equil2}\\
z_{i0}&=\xi_{h0},\;\; \forall i\in\bar{\mathcal{V}}\label{equil3}
\end{align}  
\end{subequations}
Define the  \textit{consensus subspace} as:
 \begin{equation}
 \mathcal{S}\triangleq \{\tilde{\matr{z}}\in \mathbb{R}^{n}\; |\;\tilde{\matr{z}}=\beta \cdot\matr{1}_{n\times 1}\;,\;\; \beta \in \mathbb{R} \}
 \end{equation}  
and let $\tilde{\matr{z}}_0$ denote the equilibrium point of $\tilde{\matr{z}}$. 
 Notice that, consensus between the agent state-variables $z_i,\;i\in\mathcal{V}$ can be reached, whenever the next two conditions are satisfied:
\begin{cond} 
\label{cond1}
  $\tilde{\matr{z}}_0\in\mathcal{S}$
 \end{cond}
 \vspace{-1mm}
  \begin{cond}
 \label{cond2}
  $ \tilde{\matr{z}}_0$ \text{\;\;is \textit{asymptotically stable}}
  \end{cond}\vspace{1mm}
 Observing equations \eqref{equil1}-\eqref{equil3}, we readily have that \textit{Condition}~\ref{cond1} is met with $\beta=\frac{z^*}{n}$. 
We are left to establish that \textit{Condition}~\ref{cond2} is met.  For this purpose, we define the augmented variables $y_l\triangleq (z_l-\xi_h), y_i\triangleq (z_i-\xi_h), \psi_h\triangleq(\xi_h-\xi_{h0})$ that move the equilibrium point $(\xi_{h0},\tilde{\matr{z}}^\top_0)$ to the origin. Using those, equations \eqref{consensus_equations7}-\eqref{consensus_equations9} become:
\vspace{-2mm}
\begin{subequations}
\begin{align}
\intertext{\textit{Slow sub-system}}
 \frac{d{\psi}_h}{dt}&\triangleq -n\psi_h-\sum_{i\in\mathcal{V}} y_i
 \label{consensus_equations11}\\
\intertext{\textit{Fast boundary-layer sub-system}}
\varepsilon \frac{d{y}_l}{dt}&\triangleq -y_l-\varepsilon \frac{d\psi_h}{dt} ,\;\;y_l\triangleq y_1
 \label{consensus_equations12}\\
\varepsilon \frac{d{y}_i}{dt}&\triangleq -(y_i-y_{i-1})-\varepsilon \frac{d\psi_h}{dt},\;\;\;\forall  i\in\bar{\mathcal{V}}
 \label{consensus_equations13}
\end{align}
\end{subequations}
Define the following vectors;  $\matr{y}=[y_2 \; ... \; y_i \;... \; y_{n}]^\top,\; \matr{y}\in \mathbb{R}^{n-1}$ and  $\matr{\tilde{y}}=[y_l \;\matr{y}^\top]^\top,\; \matr{\tilde{y}}\in \mathbb{R}^{n}$ and let the variable $\tau=\frac{t}{\varepsilon}$ denote the fast time-scale. On this time scale, the slow state-variable $\psi_h$ appears as ``\textit{frozen}'' and therefore we can assume that $\frac{d\psi_h}{d\tau}\approx 0$.
Under this approximation, the fast decoupled sub-system becomes:
\begin{subequations}
\begin{align}
\frac{d{y}_l}{d\tau}&=-y_l
 \label{fastman1}\\
\frac{d{y}_i}{d\tau}&=-(y_i-y_{i-1}),\;\;\;\forall  i\in\bar{\mathcal{V}}
 \label{fastman2}
\end{align}
\end{subequations}
The following lemma establishes asymptotic stability of the fast sub-system equilibrium point .
\begin{lemma}
The equilibrium point of the system  \eqref{fastman1}-\eqref{fastman2},  $(\matr{\tilde{y}}_0=\matr{0}_{n\times 1})$, is \textit{asymptotically stable}.
\label{lemmafast}
\end{lemma}
\begin{proof}
In matrix form we have:
\begin{center}
$\begin{bmatrix}
\frac{d{y}_l}{d\tau}\\
\vdots\\
\frac{d{y}_{n}}{d\tau}\\
\end{bmatrix}
$=$\underbrace{\begin{bmatrix}
-1 & 0 & \cdots & 0 & 0 \\
\vdots & \ddots  &  &   & \vdots \\
0 & 0 & \cdots & 1 & -1 \\
\end{bmatrix}}_{{\Huge \matr{A_f}}} \begin{bmatrix}
y_l\\
\vdots\\
y_{n}\\
\end{bmatrix}
$
\end{center}
Denoting the eigenvalues of $\matr{A_f}$ with $\pmb{\lambda}=[\lambda_l\;\lambda_2\;\cdots\;\lambda_{n}]^\top$ and noticing that $\matr{A_f}$ is a lower triangular matrix \cite{matrixhorn}, gives that 
 $\pmb{\lambda}=-\matr{1}_{n\times 1}$. That is equivalent to $\matr{A_f}$ being a  \textit{Hurwitz matrix} and to the origin of the system \eqref{fastman1}, \eqref{fastman2} being \textit{asymptotically stable} \cite{khalil}.
That, concludes the proof of \textit{Lemma}~\ref{lemmafast}.
 \end{proof} 
  According to \textit{Theorem 4.6} \cite{khalil} the above statements are valid if and only if for any $\matr{Q}\succ 0$, $\exists \matr{P} \succ 0$ that satisfies:
\begin{equation}
\matr{PA_f+A_f^\top P=-Q},\;\;\;\;\matr{Q},\; \matr{P}\in\mathbb{R}^{n\times n}
\label{lyapunovequation}
\end{equation}   
with stability certificate the Lyapunov function given by:
\begin{align}
V_f=\matr{\tilde{y}}^\top \matr{P} \matr{\tilde{y}}
\end{align}
 We now turn our focus on the slow time scale $t$ and notice that \textit{Lemma}~\ref{lemmafast} establishes  $\lim_{\tau\to\infty} (y_l)=0,\;\lim_{\tau\to\infty}(y_i)= 0,\; \forall i\in\mathcal{\bar{V}}$. With these, the slow sub-system \eqref{slowsubsys} can be approximated by:
\begin{align}
 \frac{d{\psi}_h}{dt}=-n\psi_h
\label{slowsubsys}
\end{align}
 
\begin{lemma}
The equilibrium point of the slow sub-system \eqref{slowsubsys}, ($\psi_{h0}=0$), is \textit{asymptotically stable}.
\label{lemmaslow}
\end{lemma}
\begin{proof}
 Take a candidate Lyapunov function as:
 \begin{equation}
 V_h=\psi_h^2,\;\;\; V_h>0, \;\;\;\forall \psi_h \in D_{\psi_h}\setminus\{0\}
\end{equation}  
Computing its time derivative along the trajectories of the system in \eqref{slowsubsys} leads to $\frac{dV_h}{dt}=-2n(\psi_h^2)$ with $n>0$ and  $\frac{dV_h}{dt}<0,\;\;\forall \psi_h \in D_{\psi_h}\setminus\{0\}$ which finally yields that $\psi_{h0}=0$ is \textit{asymptotically stable}.
\end{proof}
 \textit{Lemmas~\ref{lemmafast}, \ref{lemmaslow}} establish \textit{asymptotic stability} of the equilibrium points of the fast and the slow reduced sub-systems respectively. Nonetheless, asymptotic stability of the full system \eqref{consensus_equations11}-\eqref{consensus_equations13} equilibrium point  cannot be seamlessly inferred since the decoupled system dynamics serve as an approximation of the full system dynamics. To grapple with this challenge, we resort to a methodology, described in \cite{khalil}, whose sketch we briefly provide next. To prove asymptotic stability of the equilibrium point of the full system, a candidate Lyapunov function $V_c$ can be first formulated from the linear combination of the Lyapunov functions $V_f$ and $V_h$. Further, by deriving the time derivative of this composite candidate Lyapunov function along the trajectories of the full system we can impose negative definetess conditions in it. That, will lead a maximum bound $\varepsilon^*$ (Theorem 11.4, \cite{khalil}) on the small parameter $\varepsilon$ which can be interpreted as the maximum value of the time-scale separation ratio for which asymptotic stability of the full system is still assured.
\subsection{Delay-Independent Asymptotic Stability} 
\label{delayindepend}
In practice, when distributed protocols are implemented,  communication delays in the exchange of information arise. In this section, we study the \textit{asymptotic stability} property of the fast sub-system equilibrium point under a specific type of delays, namely the fixed-time delays. Denoting the delays with  $r\in \mathbb{R}^{+}$ we emphasize that, while the time-delays can  be arbitrary, the approximation $\frac{d\psi_h}{d\tau}\approx 0$ has to be valid so that the slow state-variable can still be considered as ``\textit{frozen}''. The time-delayed version of  \eqref{fastman1}-\eqref{fastman2} is:
\begin{equation}
\frac{\matr{d\tilde{y}}}{\matr{d}\tau}=\matr{A}_0\matr{\tilde{y}}+\matr{A}_1\matr{\tilde{y}}(\tau-r) 
\label{timedelaysystem}
\end{equation}
where  $\matr{A}_0 \triangleq -\matr{I}_{n}$ and  $\matr{A}_1$  defined as:
\begin{equation} 
 \matr{A}_1 \triangleq \begin{bmatrix}\matr{0}_{(n-1)\times 1}^\top & 0\\
  \matr{I}_{(n-1)} &\matr{0}_{(n-1)\times 1}\end{bmatrix}
  \end{equation}
We construct a \textit{Lyapunov-Krasovskii} functional \cite{timedelaykhar} as:
\begin{equation}
V_1=\matr{\tilde{y}}(\tau)^\top \matr{P_1} \matr{\tilde{y}}(\tau)+\int_{\tau-r}^{\tau}\matr{\tilde{y}}(\eta)^\top \matr{Q_1} \matr{\tilde{y}}(\eta)d\eta
\end{equation}  
where $\matr{P_1}, \matr{Q_1}\in\mathbb{R}^{n\times n}$ and $\matr{P_1}, \matr{Q_1}\succ 0$. Direct differentiation with respect to $\tau$ yields:
\begin{equation}
\frac{dV_1}{d\tau}=
{
\matr{\tilde{y}}_d^\top
\matr{\tilde{Q}_1}
\tilde{\matr{y}}_{d}
  }
 \end{equation}
 where \begin{equation}
\matr{\tilde{Q}_1}\triangleq \begin{pmatrix}
\matr{P_1A_0+A_0^\top P_1+Q_1} & \matr{P_1 A_1}\\
\matr{A_1^\top P_1} & \matr{-Q_1}
\end{pmatrix}
\label{Q1tilda}
 \end{equation}
and 
\begin{equation}
\tilde{\matr{y}}_{d}^{\top}\triangleq(\matr{\tilde{y}}(\tau)^\top\;\matr{\tilde{y}}(\tau-r)^\top),\; \tilde{\matr{y}}_{d}\in\mathbb{R}^{2n}
\end{equation}
From the \textit{Lyapunov-Krasovskii Stability Theorem} \cite{timedelaykhar}, we  have that the \textit{delay-independent asymptotic stability} of the equilibrium point is guaranteed whenever 
\begin{equation}
\frac{dV_1}{d\tau}<0,\hspace{5mm} \forall \matr{\tilde{y}}_{d}\in\mathcal{D}_{1}\setminus \{\matr{0}_{(2n\times 1)}\}
\label{v1negat}
\end{equation} 
 To prove  \eqref{v1negat}, we introduce the next Lemma.
\begin{lemma}
 $\exists \matr{P_1,Q_1} \succ 0$ diagonal matrices such that for $\matr{\tilde{Q}_1}$, given by \eqref{Q1tilda}, it holds: 
\begin{equation}\matr{\tilde{Q}_1}
  \prec 0
  \label{LMIfeas}
 \end{equation}
 \label{LMIfeaslemma}
\end{lemma} 
 \begin{proof}
Without loss of generality we take $\matr{P_1,Q_1}$ to be positive definite diagonal matrices. The  Schur complement conditions on the matrix $\matr{\tilde{Q}_1}$, give that $\matr{\tilde{Q}_1}\prec 0$ is equivalent to $\matr{P_1A_0+A_0^\top P_1+Q_1}\prec 0$ together with the Schur complement matrix (of $\matr{\tilde{Q}_1}$), $\matr{S_1}$, satisfying $\matr{S_1}\prec0$.
Performing basic matrix calculations gives:
\begin{align}
\matr{P_1A_0}&\matr{+A_0^\top P_1+Q_1}
=\nonumber\\
&\begin{bmatrix}
q_1-2p_1 & 0 &\cdots & 0\\
 0 & q_i-2p_i &\cdots & 0\\
    0 & 0 & \ddots & 0\\
 0 & 0 &\cdot & q_n-2p_n\\
\end{bmatrix}
\nonumber
\end{align}
which is negative definite ($\prec 0$) when:
\begin{subequations}
\begin{align}
q_i-2p_i<0,\;\;\; \forall i\in\mathcal{V}
\label{cond1ineq}
\end{align}
The Schur complement is obtained as:
\begin{equation}
\matr{S_1}
=\begin{bmatrix}
-q_1-\frac{p_2^2}{2p_2-q_2}& 0 &\cdots & 0\\
   0 &  -q_i-\frac{p_{i+1}^2}{2p_{i+1}-q_{i+1}} &\cdots & 0\\
   0 & 0 & \ddots & 0\\
 0 & 0 &\cdots & -q_n\\
\end{bmatrix}
\nonumber
\end{equation}
which is negative definite ($\prec 0$) when:
\begin{align}
-q_i-\frac{p_{i+1}^2}{2p_{i+1}-q_{i+1}}&<0, \;\;\;\forall i\in \mathcal{V}\setminus \{n\}
\label{cond2ineq1}\\
-q_n&<0
\label{cond2ineq2}
\intertext{Restricting $\matr{P}_1,\;\matr{Q}_1$ to be positive definite yields:}
p_i&>0\;,\;\;\; \forall i\in\mathcal{V}\label{pposd}\\
 q_i&>0\;,\;\;\; \forall i\in\mathcal{V}\label{qposd}
\end{align}
\end{subequations}
Now, we have to prove that we can find $p_i,q_i\;,\forall i\in\mathcal{V}$, that satisfy the inequalities \eqref{cond1ineq}-\eqref{qposd} i.e LMI feasibility. Our proof for the LMI feasibility is constructive; we introduce a methodology for finding $p_i,q_i$ to meet the inequalities. We begin by choosing  $p_i,\;\forall i\in\mathcal{V}$ s.t \eqref{pposd} holds. Next, choose the $q_n$ s.t \eqref{qposd},\eqref{cond1ineq} together hold and each $q_i,\;\forall i\in \mathcal{V} \setminus\{n\}$ s.t both \eqref{cond1ineq}, \eqref{cond2ineq1} are satisfied. Using this methodology, we can always find $p_i,q_i\;,\forall i\in\mathcal{V}$ that meet the inequalities \eqref{cond1ineq}-\eqref{qposd}. Hence, the LMI problem in \eqref{LMIfeas} is indeed feasible and we conclude the proof. \end{proof} 

 \begin{theorem} The equilibrium point ($\matr{\tilde{y}}_0=\matr{0}_{n\times 1}$) of the system \eqref{timedelaysystem} is  delay-independent asymptotically stable.
 \label{thmtimedelays}
 \end{theorem} 
\begin{proof}
From Lemma~\ref{LMIfeaslemma} and \textit{Lyapunov-Krasovskii Stability Theorem} \cite{timedelaykhar} we  conclude the proof.
\end{proof}

Intuitively, Theorem~\ref{thmtimedelays} guarantees that the fast consensus state-variables $z_i,\;i\in\mathcal{V}$ converge to the slow state-variable $\xi_h$ independently of any time-delays, as long as $\psi_h$ remains ``\textit{frozen}''. That is of practical interest, since it certifies that the performance of the protocol $\mathcal{P}_1$ is robust with respect to  time-delays that are inherent in communication channels.

\begin{figure}
       \begin{subfigure}{0.37\textwidth}
   \includegraphics[scale=0.36]{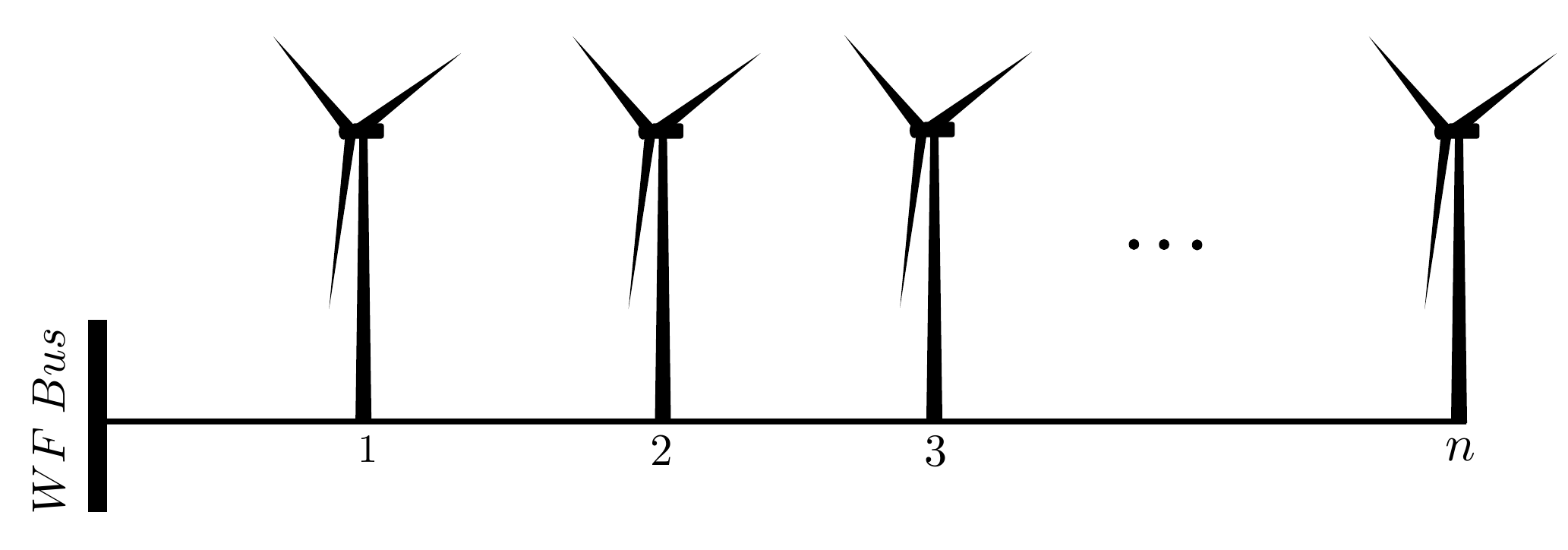}
        \caption{}
        \label{physicalcon}
    \end{subfigure}\\
    \begin{subfigure}{0.36\textwidth}    
   \includegraphics[scale=0.34]{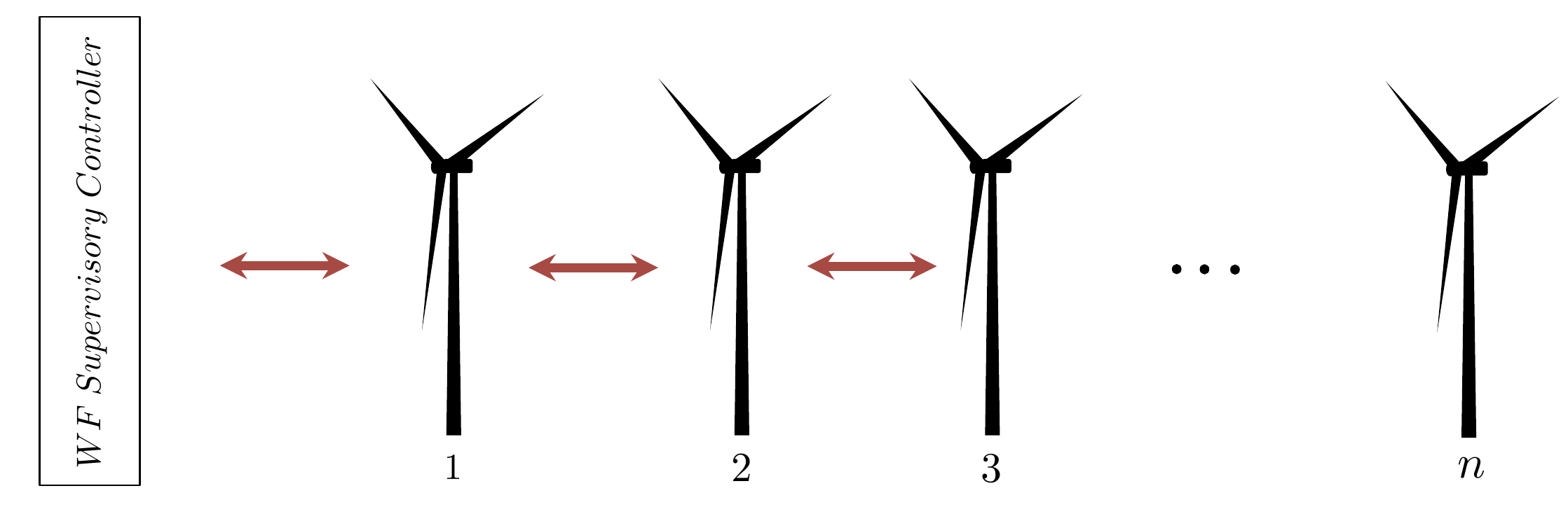}
        \caption{} 
        \label{commcon}
    \end{subfigure}
        \label{physcomtop}
         \begin{subfigure}{0.55\textwidth}    
   \includegraphics[scale=0.50]{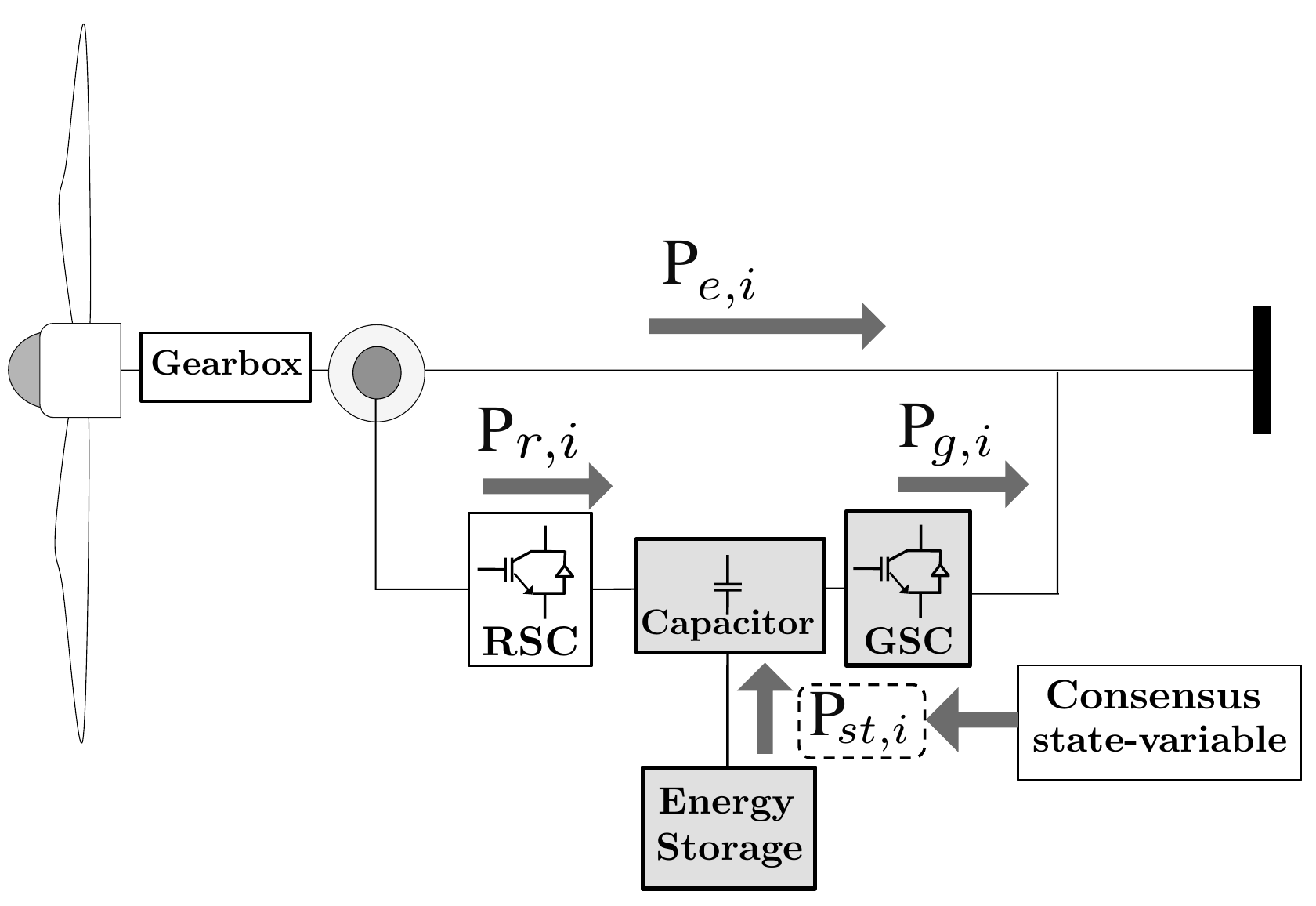}
        \caption{} 
           \label{DFIGwithstor}
    \end{subfigure}
         \caption{a) Physical topology of the WF\; b) Communication links \; c) WG with integrated storage}
        \label{physcomtop}
        \vspace{-7.5mm}
\end{figure}

\section{WF Short-term Predictable Power Via Distributed Control of the WGs}

In this Section, we apply the proposed protocol for coordinating a group of WGs storage devices in a WF toward achieving a common objective. The group objective is short-term predictable power from the WF i.e provide constant power or track a reference with load sharing between the  WGs.
Let the set $\mathcal{G}\triangleq\{1,...,n\}$ denote the WGs with integrated storage.  To ease the notation, we simply denote  a wind generator of this type by WG and index it by  $i$, where $i\in\mathcal{G}$. Without loss of generality we take $l\triangleq 1$. Let $\bar{\mathcal{G}}$ denote the set of WGs without the leader WG.  In our case, $\bar{\mathcal{G}}\triangleq\{2,...,n\}$.  In addition, we assume that each WG is an agent and we have WGs exchanging information between each other.  Let $P_{e,i}$, $P_{g,i}$, $P_{r,i}$ and $P_{st,i}$ denote the electric power produced by the stator, the grid-side converter (GSC), the rotor-side converter (RSC) and the storage respectively. All these can be seen in Fig.\ref{DFIGwithstor}. In general, the WFs obtain a reference power output, $P_d$ from a system operator and this corresponds to the WF committed output. The committed power output is obtained as an outcome of a forecasting method and an Economic Dispatch process, conducted by the system operator taking into account local wind conditions. 
 Consider the dynamics of the capacitor that interfaces the Rotor Side Converter (RSC) and the GSC of each WG, as (Fig.\ref{DFIGwithstor}):
\begin{equation}
(C_{dc,i} V_{dc,i}) \frac{dV_{dc,i}}{dt}\triangleq (P_{r,i}+P_{st,i}-P_{g,i}),\; i\in \mathcal{G}
\label{capacitordyn}
\end{equation}
where $C_{dc,i},V_{dc,i}\in\mathbb{R}$ are the capacitance and the voltage potential of the capacitor respectively. When the storage is inactive, at the equilibrium we have:
\begin{equation}
P_{g,i}=P_{r,i},\;\; i\in \mathcal{G}
\end{equation}
Taking that into account,   we have that the total WF power output is equal to the total available mechanical power that is coming from the wind i.e:
\begin{equation}
\sum\limits_{i\in\mathcal{G}} P_{m,i}\approx\sum\limits_{i\in\mathcal{G}} (P_{e,i}+P_{r,i})
\end{equation}
 The total mechanical power  $\sum_{i\in\mathcal{G}} P_{m,i}$ is highly variable and since it depends on the wind speed conditions, it varies minute by minute. As a result, for a WF we might have that:
\begin{equation}
\sum\limits_{i\in\mathcal{G}} P_{m,i}\approx\sum\limits_{i\in\mathcal{G}} (P_{e,i}+P_{r,i})<P_d
\end{equation} 
  The opposite case, where more wind power is available does not pose any problem since the WF can ``\textit{spill}'' wind. On the other hand, the reasons for which the WF cannot generate  $P_d$ are twofold. Firstly, due to time delays between the time that the system operator conducts the wind forecasting and defines  $P_d$ until the time that this set-point is communicated to the WF. In this case, the WF available mean power can actually be less than $P_d$, making it not possible for the WF to generate the required amount of power. Secondly, due to the continuous wind fluctuations, the WF can only generate $P_d$ on average. Although the WF cannot always extract (from the wind) power equal to $P_d$, it can generate a total power $P_d$ when the WGs have storage devices, able to contribute. When the storage is active i.e $P_{st,i}>0$ or $P_{st,i}<0$, at the equilibrium of \eqref{capacitordyn}, we get:
\begin{equation}
P_{st,i}=(P_{g,i}-P_{r,i}),\;\; i\in \mathcal{G}
\label{storagepower}
\end{equation} 
The additional flexibility offered by the integrated storage can enable the WGs to actually generate the required $P_d$, given that the storage devices have stored enough energy. Hence, the WGs can provide short-term predictable power and balance wind intermittency as well. With storage, we have: 
 \begin{equation}
  \sum\limits_{i\in\mathcal{G}} (P_{e,i}+P_{r,i}+P_{st,i})= P_d 
  \label{storeq}
 \end{equation}
  The storage devices can provide or draw the excess power that is required by the WF to meet its committed power output in a short time-scale. That, can be translated into the next condition for the total storage power.
  \begin{cond}
 \begin{equation}
 \lim_{t\to\infty}  \sum\limits_{i\in\mathcal{G}} P_{st,i}= P_d -\sum\limits_{i\in\mathcal{G}} (P_{e,i}+P_{r,i})
   \label{storagemismatch1}
 \end{equation}
 \label{condsumpower}
 \end{cond}
 For the scope of utilizing the available resources efficiently,  the storage devices can be controlled such that they contribute equally to the excess power required by the WF to reach $P_d$. That, is translated into the following condition.
 \begin{cond}
 \label{conspowers}
 \begin{equation}
\lim\limits_{t\to\infty} P_{st,i}=\lim\limits_{t\to\infty} P_{st,j}\;,\hspace{5mm}\forall i,j\in \mathcal{G}
\label{pstequal}
\end{equation}
 \end{cond}
 \vspace{2.5mm}
 Consider the following definition.  
\begin{mydef}
\textnormal{Any utilization scenario between the energy storage devices of the WGs for which  \textit{Conditions}~\ref{condsumpower},\ref{conspowers} hold},  \textit{is called a fair utilization scenario}. 
\label{fairutil}
\end{mydef}
The problem we seek to address can now be formulated as.
\begin{problem}
\textit{To coordinate the energy storages of various state-of-the-art WGs in a completely distributed way i.e each WG  exchanges information with only its neighbor-WGs, such that \textit{Conditions}~\ref{condsumpower},\ref{conspowers} hold.}
\label{Problem1}
\end{problem}
Next, we introduce a methodology based on the \textit{Leader-follower Consensus Protocol $\mathcal{P}_1$} to address Problem~\ref{Problem1}.

\subsection{Proposed Methodology}
In this section, we propose a methodology that addresses \textit{Problem}~\ref{Problem1}. In our set-up, the physical and the communication topology of the WF is as shown in Fig.~\ref{physicalcon},\;\ref{commcon}. Notice that, the communication topology is exactly identical to the one in Fig.~\ref{nodegraph} with each WG exchanging information with one neighboring WG. Also, the physical topology is in alignment with the communication topology.
Using Eq.~ \eqref{storagepower}, Condition~\ref{conspowers} yields to: 
\begin{cond}
\label{zicondition}
\begin{equation}
\lim_{t\to\infty} z_i=\lim_{t\to\infty}z_j,\; i,j \in \mathcal{G}
\label{consens}
\end{equation}
\end{cond}
where  $z_i\triangleq (P_{g,i}-P_{r,i}),\;\;\forall i\in\mathcal{G},\;z_i\in\mathbb{R}$. It springs from \eqref{consens} that, we can pose Problem \ref{Problem1} as a constrained consensus agreement problem. Further, we can address this problem by requiring each WG  to exhange its $z_i$ information with one of its neighbors according to  Fig.\ref{commcon}, until they reach consensus on their $z_i's$. In this case,  we would ensure that Eq.~\eqref{consens}  holds. To realize this process, we propose the  protocol given below.
\begin{mydef}[Protocol $\mathcal{P}_2$] 
\begin{subequations}
\begin{align}
\intertext{\textit{Leader WG}}
\frac{d\xi_h}{dt}&\triangleq\Big(P_d-\sum_{i\in \mathcal{G}} (P_{e,i}+z_i+P_{r,i})\Big),\hspace{4mm} \xi_h\in\mathbb{R}\label{consensus_equations_p2_1}\\
\frac{dz_l}{dt}&\triangleq-k_{\alpha,l}(z_l-\xi_h),\hspace{27mm} z_l\in\mathbb{R}\label{consensus_equations_p2_2}\\
 z_l&\triangleq z_1\nonumber\\
\intertext{\textit{WG} $i$}
\frac{dz_i}{dt}&\triangleq-k_{\alpha,i}(z_i-z_{i-1})\;,\;\;\hspace{10mm} z_i\in\mathbb{R},\;i\in\bar{\mathcal{G}} \label{consensus_equations_p2_3}
\end{align}
\end{subequations}
\end{mydef}
 where $\xi_h$ is the additional state-variable of the leader WG.  
 \begin{assumption}
\textit{The information $\sum_{i\in\bar{\mathcal{G}}}(P_{e,i}+z_i+P_{g,i})$ can be retrieved by the leader WG.} 
 \label{pmtotal}
 \end{assumption}
 By defining $z^*=\Big(P_d-\sum_{i\in \mathcal{G}} (P_{e,i}+P_{r,i})\Big)$, the Protocol  $\mathcal{P}_2$ takes the form of the Protocol $\mathcal{P}_1$.	 Thus, the preceding stability analysis of $\mathcal{P}_1$ still applies. Notice that, $\mathcal{P}_2$ guarantees that the total power output of the WF will be $P_d$ since at the equilibrium, equation $\frac{d\xi_h}{dt}=0$ is identical to \eqref{storeq}. Summarizing, the Protocol $\mathcal{P}_2$ effectively addresses the Problem~\ref{Problem1}.
To implement $\mathcal{P}_2$ the GSC  and the storage controllers have to be designed appropriately Fig.\ref{DFIGwithstor}. We briefly describe how this can be done ommiting the full derivation due to space limitation. The GSC controller for the leader and the rest of the WGs can be designed using equations \eqref{consensus_equations_p2_2} and \eqref{consensus_equations_p2_3}.  Besides that, the storage controllers are critical to be designed such that by the time the GSCs reach consensus on the $z_i's$, for the storage power, it holds:
\begin{equation}
x_i=z_i,\;\forall i\in \mathcal{G}
\end{equation}
where $x_i=P_{st,i},\;\forall i$. A storage controller that achieves this objective can be found in our previous work, \cite{ectropywindstef}, \cite{Baros}. As a general guideline, we emphasize that the storage controller has to guarantee that $x_i$ is tracking the consensus state $z_i$, regulated by the GSC, in a much faster time-scale than the time-scale of the consensus protocol. This is very important for having:
\begin{center}
\textit{Condition}~\ref{zicondition} is true $\Rightarrow$
\textit{Condition}~\ref{conspowers} is true\\
\end{center}
That, results to guaranteed consensus on $P_{st,i},\;\forall i\in\mathcal{G}$, whenever consensus on the  $z_i,\;\forall i\in\mathcal{G}$ is reached.

\begin{figure}
        \centering
\includegraphics[scale=0.45]{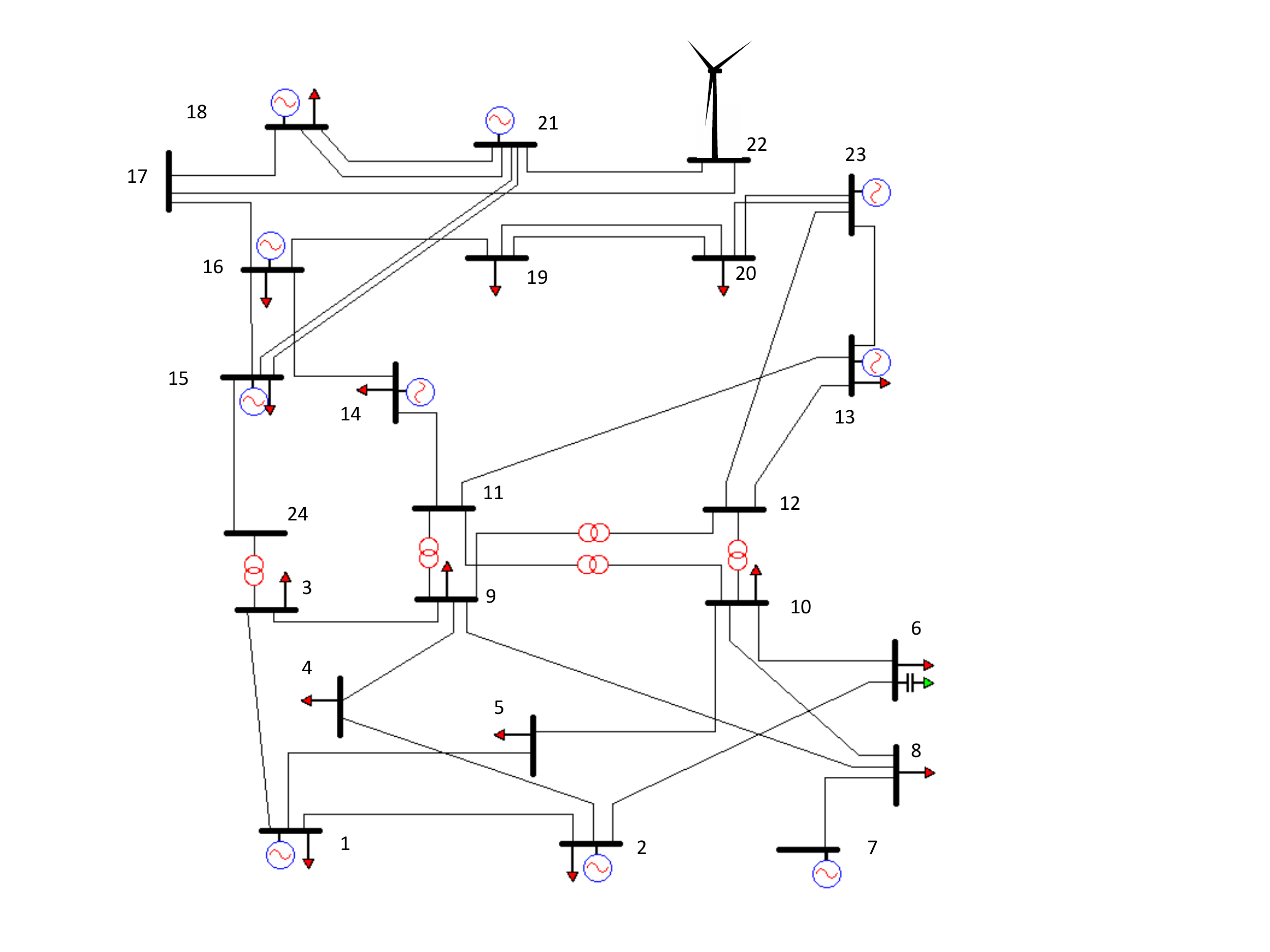}
        \caption{Modified IEEE 24-bus RT system}
        \label{24buswind}
        \vspace{0mm}
\end{figure}%

\begin{figure}
        \centering
       \begin{subfigure}{0.4\textwidth}
\includegraphics[scale=0.37]{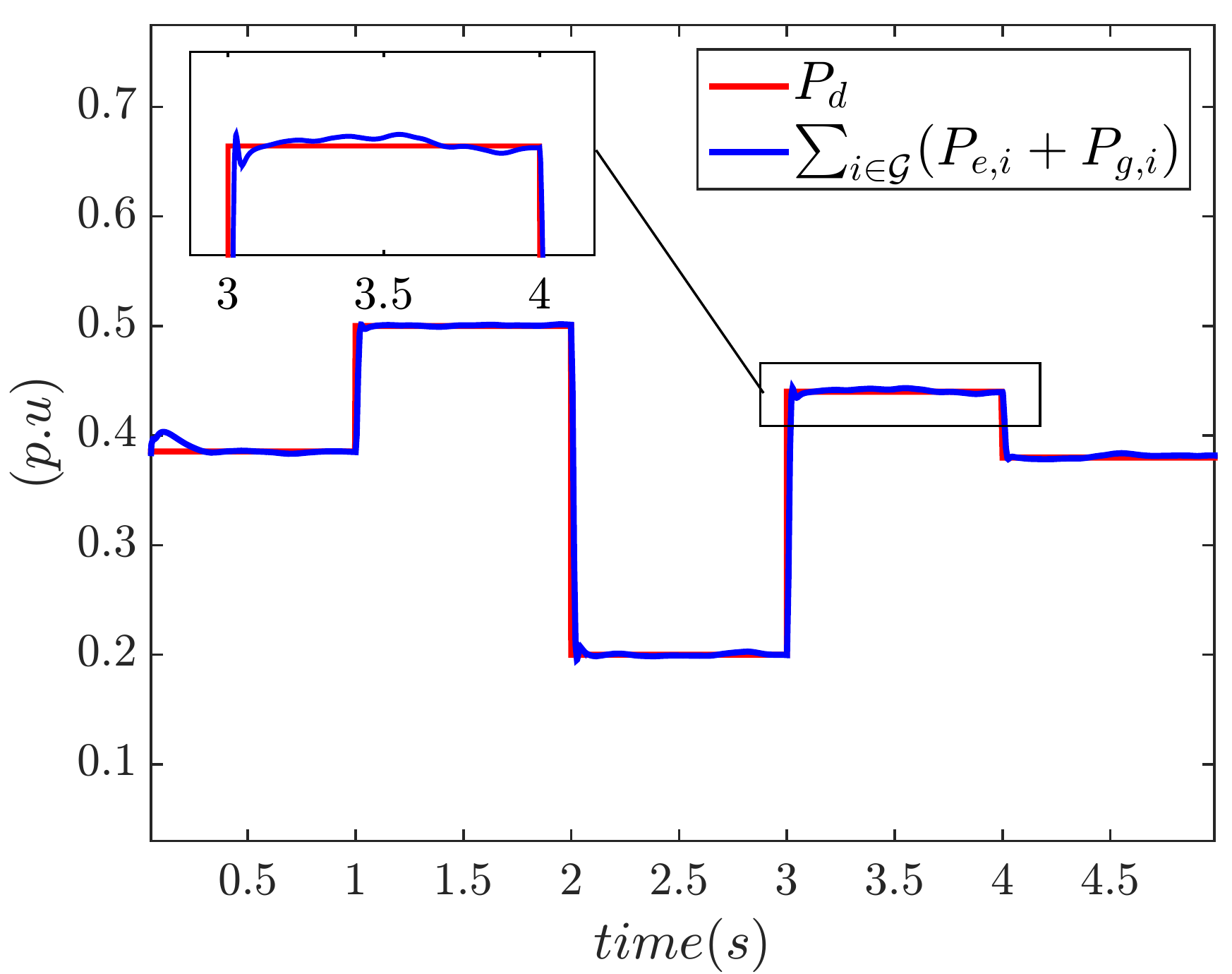}
        \caption{{\footnotesize Total WF power output}}      \vspace{0.5mm}
        \label{scen1_1}
    \end{subfigure}\\
    \begin{subfigure}{0.4\textwidth}    
\includegraphics[scale=0.367]{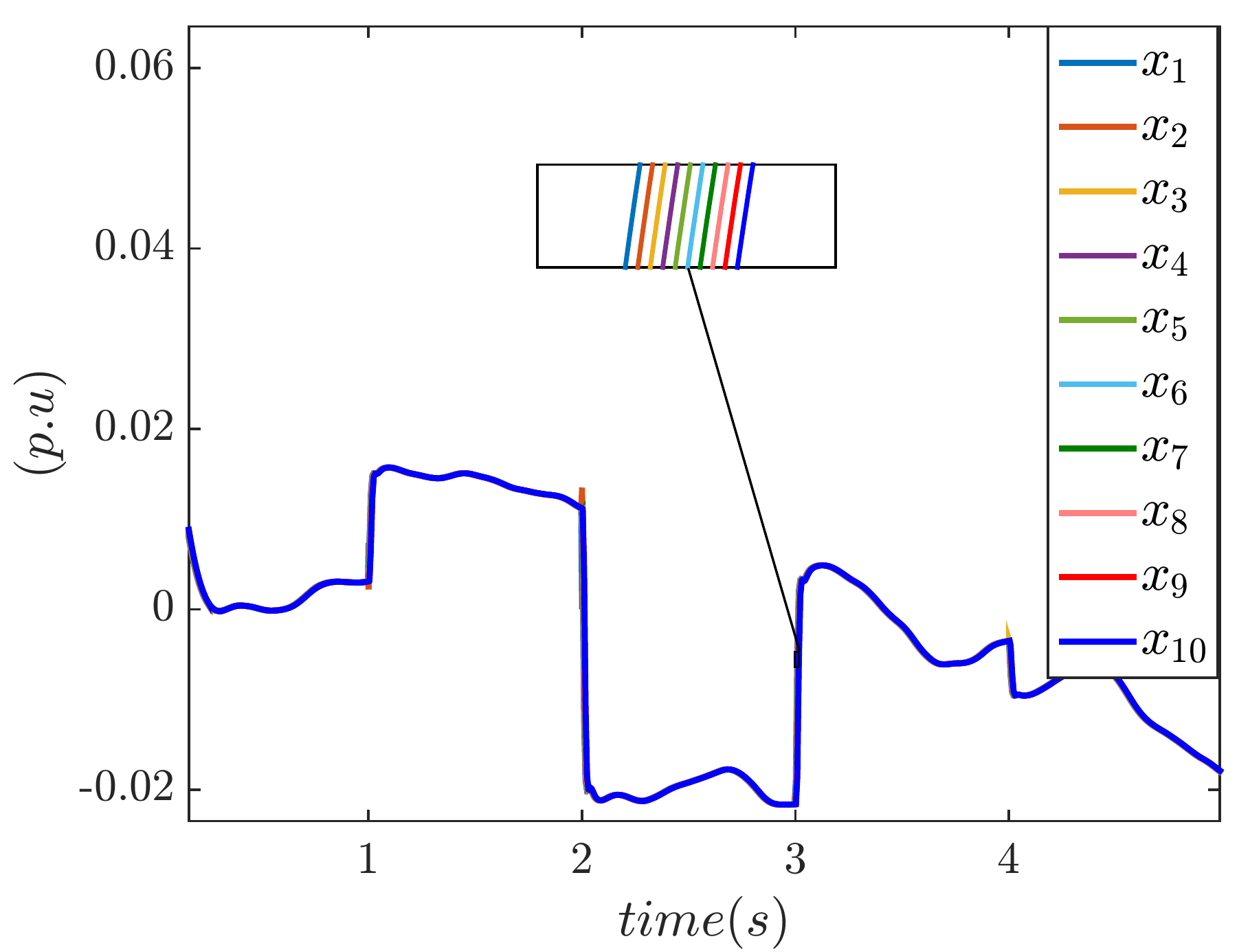}
        \caption{{\footnotesize Storage power of each WG, $P_{st,i}$}}       \vspace{0.5mm}
        \label{scen1_2}
    \end{subfigure}\\
        \begin{subfigure}{0.4\textwidth}    
\includegraphics[scale=0.37]{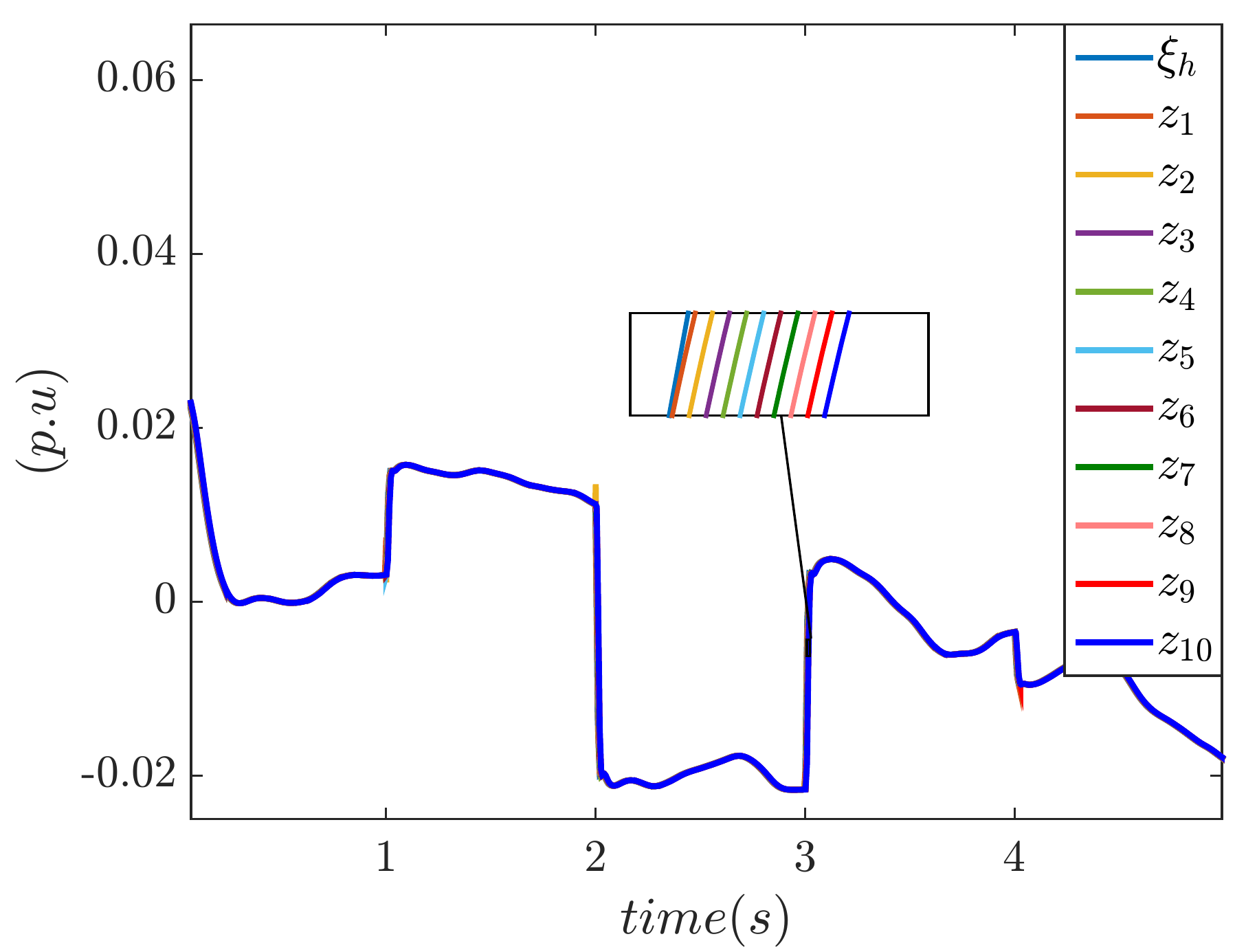}
        \caption{{\footnotesize Consensus state of each  WG $z_{i}$}}       \vspace{0.5mm}
        \label{scen1_3}
    \end{subfigure}\\
    \caption{WF response under \textit{Scenario 1 }}
        \label{Scen1}
        \vspace{-6.5mm}
\end{figure}
\begin{figure}
    \centering
    \begin{subfigure}[b]{0.42\textwidth}
        \centering
\includegraphics[scale=0.36]{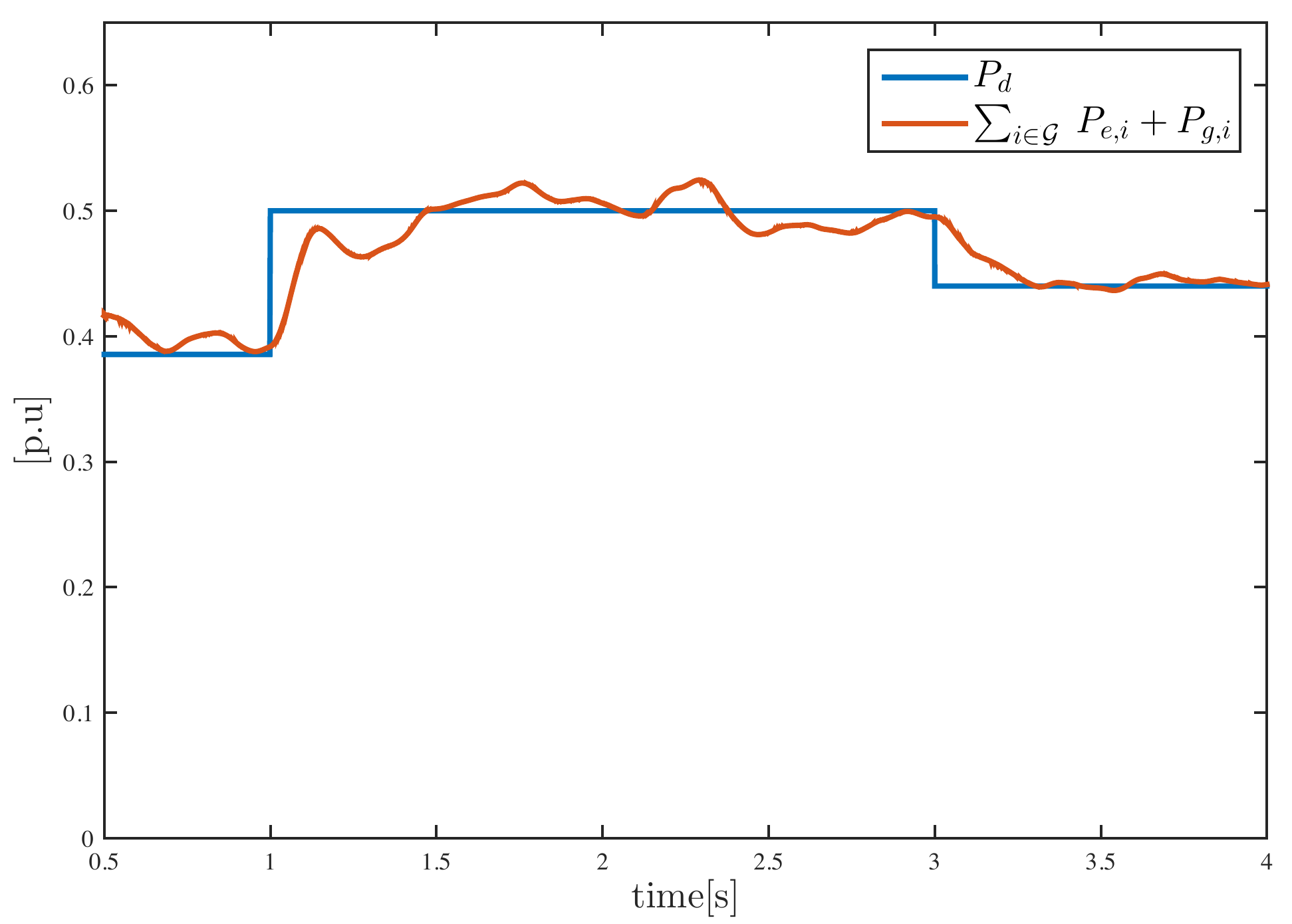}
        \caption{{\footnotesize Total WF power output}}
        \vspace{0.5mm}
        \label{scen3_1}
    \end{subfigure}%
    
    \begin{subfigure}[b]{0.4\textwidth}
        \centering
\includegraphics[scale=0.41]{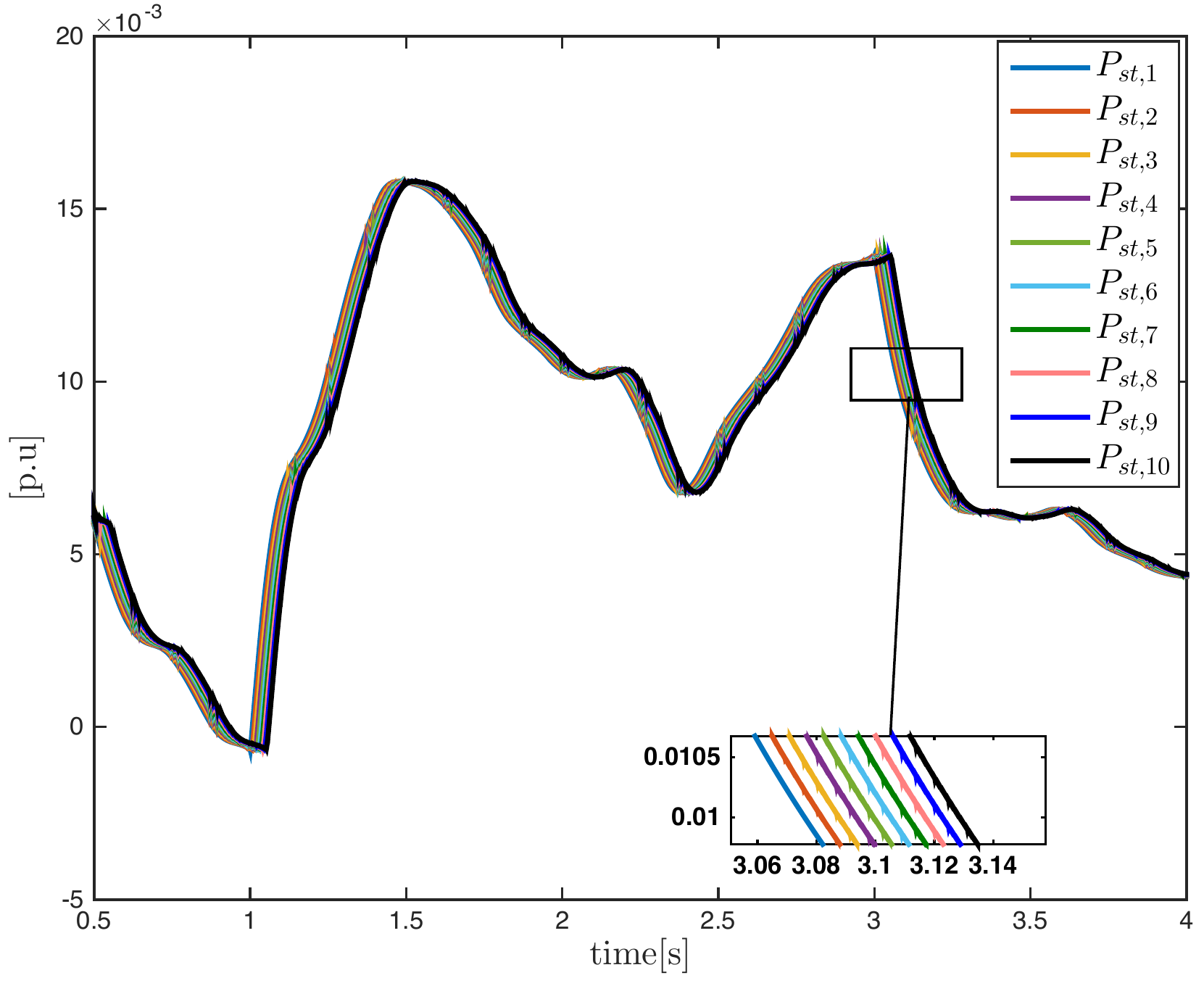}
        \caption{{\footnotesize Storage power of each WG, $P_{st,i}$}} 
        \vspace{0.3mm}
        \label{scen3_2}
    \end{subfigure}
     \begin{subfigure}[b]{0.4\textwidth}
        \centering
\includegraphics[scale=0.41]{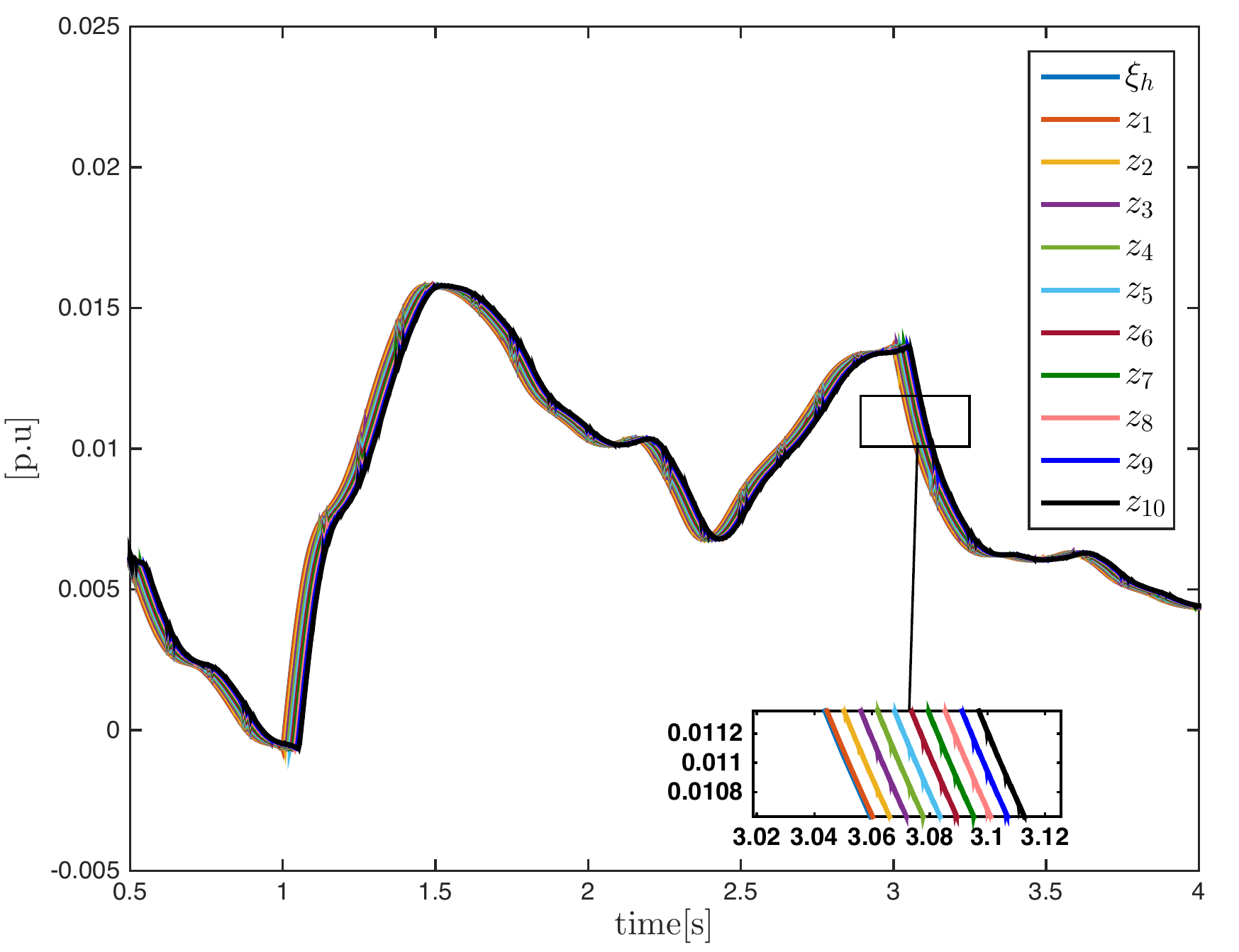}
        \caption{{\footnotesize Consensus state of each  WG $z_{i}$}}
        \label{scen3_3}
    \end{subfigure}
    \caption{WF response under \textit{Scenario 2}}
    \label{Scen3}
    \vspace{-3.5mm}
\end{figure}  

  \section{Performance Evaluation}
To validate our results, we apply our methodology on the modified IEEE 24-bus RT system. The goal is the WF power output to track a varying reference when is suddenly ordered to and when the available wind power is not  sufficient. Also, the storage power outputs are desired to contribute to this power mismatch in an equal-sharing manner. At bus 22 of the 24-bus power grid, we have a WF comprised of 10 WGs where each of them has a storage encapsulated into its scheme Fig.\ref{DFIGwithstor}. The physical as well as the communication structure of the WF is the same as in Fig.\ref{physicalcon}, \ref{commcon}. We verify our proposed methodology, by conducting simulations under the next 2 critical scenarios. 
\begin{itemize}
\item \textit{\textbf{Scenario 1}:} \textit{Step-wise changes on the reference $P_d$ and no time-delays.}
\item \textit{\textbf{Scenario 2}:} \textit{Step-wise changes on the reference $P_d$ and fixed delays $r=5ms$.}
\vspace{2mm}
\end{itemize}
Under \textit{Scenario 1}, the WF dynamics evolve as shown in Fig.~\ref{Scen1}. More specifically,  we observe from Fig.~\ref{scen1_1} that the total power of the WF is  tracking the reference  with good transient response i.e no oscillations, no overshoot. Hence, the WF total power output is reaching the committed reference  rapidly. Additionally, from Fig.~\ref{scen1_2} we can see that the power coming from each storage is identical. This, verifies that the storage power outputs reached consensus.  Direct comparison of the responses  in Fig.~\ref{scen1_3} and Fig.~\ref{scen1_2} gives that, the storage power outputs and the consensus states are one-to-one identical to each other. 
This, verifies the effectiveness of the time-scale-separation-based design of the GSC and storage controllers which we explain as follows. The storage devices continuously generate the  power that is required by the GSCs, whereas, the GSCs via the consensus protocol, regulate their power to match the total WF power, $P_d$. In the second scenario, we regarded time-delays on the exchange of information between the WGs. The WF response under this scenario can be seen in Fig.~\ref{Scen3}.    It is obvious that the time-delays deteriorated the tracking response,  Fig.\ref{scen3_1}, but the control objectives were still achieved. Specifically, the WF power output was tracking a reference and  consensus on the power output of each storage device Fig.\ref{scen3_2}, was attained. These results verify the time-delay independent stability property of the proposed protocol, proved in Section~\ref{delayindepend}. 
   We conclude from the above results that, our proposed methodology is able to coordinate a fleet of WGs such that the total power from a WF is tracking a reference  while using the available storage devices efficiently i.e addressing Problem~\ref{Problem1}.

\section{Concluding Remarks and Future Work}
In this paper, we introduced a \textit{Leader-follower Consensus Protocol} under a specific communication set-up. We first analyzed and proved its stability properties using singular perturbation arguments. Further, we extended these results by proving that its stability property is not altered under time-delays i.e its stability property is delay-independent. On the practical side, we showed how this protocol can find applications in power grids by proposing a methodology that can be adopted by WFs to address an emerging challenge. In particular the challenge is to coordinate, through appropriate communication and control, the storage devices of a fleet of state-of-the-art WGs toward achieving a group objective. In our case the group objective was the WF power output to track a reference using the storage devices contributing in a fairly manner.   Finally, we used the IEEE 24-bus RT system to verify our theoretical results and to demonstrate the effectiveness of our proposed methodology. In future work, we would like to explore how our methodology can be advanced such that the WFs can provide more services to the grid e.g frequency and inertial response, with efficient coordination of its available WGs.

\nocite{*}
\IEEEpeerreviewmaketitle
\bibliographystyle{unsrt}
\bibliography{Distr_cont_of_DFIGs_wstor_ACC2016}{}

\end{document}